\documentclass[11pt, a4paper]{amsart}
\usepackage[dvips]{epsfig}
\usepackage{amsmath}
\usepackage{amssymb}
\usepackage{amsfonts}
\usepackage{amsthm}
\usepackage{amsbsy}
\usepackage{amsgen}
\usepackage{amscd}
\usepackage{amsopn}
\usepackage{amstext}
\usepackage{amsxtra}

\newtheorem{theorem}{Theorem}[section]
\newtheorem{proposition}[theorem]{Proposition}
\newtheorem{lemma}[theorem]{Lemma}
\newtheorem{corollary}[theorem]{Corollary}

\theoremstyle{definition}
\newtheorem{definition}{Definition}
\theoremstyle{remark}
\newtheorem{remark}[theorem]{Remark}

\newcommand {\R} {\mathbb{R}}

\newcommand {\E} {\mathbb{E}}

\newcommand {\C} {\mathbb {C}}

\newcommand {\ZL} {{\mathcal Z}^{L}}
\newcommand {\ZS} {{\mathcal Z}^{S}}

\newcommand {\Z} {{\mathcal Z}}
\newcommand {\N} {{\mathcal N}}
\newcommand {\intersect}{\mathcal{I}}
\newcommand {\numzeros} {\tilde{\mathcal I}}

\title[Counting nodal lines]{Counting open nodal lines of random waves on planar domains}
\author{John A. Toth and Igor Wigman}

\address{Department of Mathematics and
Statistics, McGill University, 805 Sherbrooke Str. West, Montr\'eal
QC H3A 2K6, Ca\-na\-da.} \email{jtoth@math.mcgill.ca} \thanks{The first author was partially supported by NSERC grant \# OG0170280 and a William Dawson Fellowship \\ The second author is supported by a CRM-ISM Fellowship}

\address{Centre de recherches math\'ematiques (CRM),
Universit\'e de Montr\'eal C.P. 6128, succ. centre-ville Montr\'eal,
Qu\'ebec H3C 3J7, Canada}
\email{wigman@crm.umontreal.ca}
\begin{document} \maketitle

\begin{abstract}

We compute the asymptotic expectation of the number of open nodal
lines for random waves on smooth planar domains. We find that for both the long energy window $[0,\lambda]$, and the short
one $[\lambda,\lambda+1]$, the expected number of open nodal lines is proportional to $\lambda,$ asymptotically as $\lambda \rightarrow \infty.$  Our results are consistent with the predictions of Blum, Gnutzmann and Smilansky \cite{BGS} in the physics
literature.
\end{abstract}

\section{Introduction}

Let $\Omega \subset {\mathbb R}^{2}$ be a smooth planar domain and consider the Dirichlet problem
\begin{equation} \label{dirichlet}
-\Delta \varphi_{j} = \lambda_j^{2} \varphi_j;
\,\,\,\,\varphi_j|_{\partial \Omega} = 0.
\end{equation}
Here, we assume that $\int_{\Omega} \varphi_i \varphi_j dx = \delta_j^{i}$. Consider the zero set of the $j$-th Dirichlet
eigenfunction (with the boundary excluded)
$${\mathcal N}_{j} := \{ x \in \Omega; \varphi_{j}(x) = 0 \} - \partial \Omega.$$
 We call the set ${\mathcal N}_{j}$  the {\em nodal set} of the eigenfunction $\phi_j$. This set is a curve which in general has self intersections. However, for generic domains \cite{U}, a nodal set is a union of connected components consisting of closed loops
homeomorphic to circles and {\em open nodal lines} homeomorphic to open
intervals.  In this paper, we compute  the
average number of open nodal lines  for a  random linear combination of
 Dirichlet eigenfunctions in various spectral intervals. This problem is generically equivalent to counting the zeros of the boundary traces of the eigenfunctions.
 The latter problem has obvious similarities
to measuring the length of the interior nodal line, and our
results show that the order of magnitude is the same. We recall
that  S. T. Yau conjectured that in all dimensions, the
hypersurface volume should satisfy  $c \lambda_j \leq
\mathcal{H}^{n-1}(\mathcal{N}_j) \leq C \lambda_j$ for some positive
constants $c, C$ depending only on $(M, g)$ \cite{Y1,Y2}. The
lower bound was proved in dimension two for smooth domains by
Br\"uning-Gomes  \cite{BG} and both the upper and lower bounds
were proved in all dimensions for analytic $(M, g)$ by
Donnelly-Fefferman \cite{DF,DF2}.  Recently, there have been important
advances in understanding other aspects of the geometry of  nodal sets for large
eigenvalues, including asymptotic results for the expected number of
nodal domains for random spherical harmonics \cite{NS}, as well as
for the distribution of nodal sets on general Riemannian manifolds
without boundary \cite{Z2}.

To state our results, we define the random linear combination of
eigenfunctions (i.e. {\em random waves}) corresponding to the long range energy window
\begin{equation}
\label{eq:PhiaL def} \Phi_{a}^{L}(x;\lambda):= \sum_{\lambda_j \in
[0,\lambda]} a_j \varphi_{j}(x),
\end{equation}
and the random combination corresponding to the short range energy
window
\begin{equation}
\label{eq:PhiaS def} \Phi_{a}^{S}(x;\lambda) = \sum_{\lambda_j \in
[\lambda, \lambda +1]} a_{j} \varphi_{j}(x),
\end{equation}
where in both cases, $a_{j}$ are $(0,1)$ Gaussian i.i.d defined on
the sample space
$(\R^{N^{L,S}(\lambda)};e^{-\|a\|^2/2}\frac{da}{(2\pi)^{N^{L,S}(\lambda)/2}
})$ and
\begin{equation}
\label{eq:spec fnc long short def} N^{L}(\lambda)
:=\#\{\lambda_{k}\in [0,\lambda] \} ;\; \; N^{S}(\lambda):=
\#\{\lambda_{k}\in [\lambda,\lambda+1] \} .
\end{equation}

The nodal set of the random wave $\Phi_{a}^{L,S}$ is by definition the curve
\begin{equation*}
\mathcal{N}_{a}^{L,S}(\lambda) = \{  x\in \Omega ; \Phi^{L,S}_{a}(x;\lambda)
= 0 \}- \partial \Omega.
\end{equation*}
We are interested here in computing the asymptotics of the number of
intersections of the nodal set with the boundary, $\partial
\Omega$. Let
$$\intersect^{L,S}_{a}(\lambda) = \text{card} \,  ( \bar{\mathcal{N}}_{a}^{L,S} (\lambda) \cap \partial \Omega ).$$
Since for generic domains, a nodal set is a union of closed loops
homeomorphic to circles and open nodal lines homeomorphic to open
intervals, the
intersection number $\intersect^{L,S}_{a}(\lambda)$, is therefore,
almost surely, equal to twice the number of open nodal lines of the
nodal set, $\mathcal{N}^{L,S}_{a}(\lambda)$; with probability $1$,
there are no multiple intersections with the boundary (see Lemma
\ref{lem:dbl zer prob 0}).

Given the random variables, $\intersect^{L,S}_{a}(\lambda)$, we
define the corresponding expectations
\begin{equation*}
\Z^{L,S}(\lambda) = \E\intersect^{L,S}_{a}(\lambda).
\end{equation*}
Our main result is the computation of the leading asymptotics of
$\Z^{L,S}(\lambda)$ as $\lambda \rightarrow \infty$.  For the following theorem, we need to impose a generic non-recurrence condition. We say that a point $q \in \partial \Omega$ is {\em non-recurrent} if the measure of loops at $q \in \partial \Omega$ for the associated billiard map $\beta: B^*\partial \Omega \rightarrow B^{*} \partial \Omega$  is zero. In the following, we denote the arclength of $\partial \Omega$ by $\ell(\partial \Omega)$.
\begin{theorem} \label{mainthm}
Let $\varphi_{j}; j=1,2,\ldots$ be the Dirichlet eigenfunctions of a
smooth domain $\Omega$ and asssume that all boundary points are non-recurrent.  Then, given $a_{j}; j=1,\ldots,N^{L,S}(\lambda),$  i.i.d.
$(0,1)$-Gaussian random variables, $$(i) \,\, \ZL(\lambda) =  \frac{\ell(\partial \Omega)}{\sqrt{6} \pi} \, \lambda + o(\lambda), $$  and

$$(ii) \,\, \ZS(\lambda)  = \frac{\ell(\partial \Omega)}{2\pi} \, \lambda + o(\lambda) $$
as $\lambda\rightarrow\infty$.
\end{theorem}
\begin{remark} The non-recurrence assumption in Theorem \ref{mainthm} holds for generic boundaries and, in particular, for all convex analytic domains \cite{Z1}.
\end{remark}

The result in Theorem \ref{mainthm} is related to well-known results
of Berard \cite{Be} on the expected nodal lengths of random
superpositions of eigenfunctions and Berry \cite{Berry 2002}, who
computed the expected length of nodal lines for isotropic,
monochromatic random waves in the plane (eigenfunctions of the  free
Laplacian). He also gave a somewhat heuristic argument for computing
the asymptotics of the variance.  Of more direct relevance is the recent result by Zelditch \cite{Z2} on the expected nodal distribution of random waves on compact manifolds without boundary. One can naturally view
Theorem \ref{mainthm} as the analogue  for domains with boundary.
Blum, Gnutzmann and Smilansky ~\cite{BGS}\footnote{We  thank
Zeev Rudnick for pointing out this reference} have studied the
distribution of the number of boundary intersections
$\tilde{\nu_{j}}$ of the nodal set of $\varphi_{j}$ (in addition to
the number of nodal domains). Using Berry's random wave to model
$\phi_{j}$ for chaotic systems, they found that for large eigenvalues, $\tilde{\nu_{j}}$ should concentrate around
$\frac{\ell}{2\pi}\cdot \lambda_{j}$, where $\ell$ is the length of the
boundary $\partial\Omega$. Numerical results for eigenfunctions of both Sinai and stadium billiards support this prediction. Since for general domains, the average nodal asymptotics over spectral intervals $[\lambda, \lambda +1]$ should be the same as for individual eigenfunctions of ergodic billiards, our asymptotic result for ${\mathcal Z}^{S}(\lambda)$ in Theorem \ref{mainthm} is thus consistent with \cite{BGS}.

Recently, in the case of piecewise-analytic domains, Toth and Zelditch \cite{TZ1} have proved deterministic upper bounds for the
number of nodal intersections with the boundary $\partial \Omega$
(and more general interior curves) for {\em individual }
eigenfunctions. In work in progress, when  $\Omega$ is an ergodic
billiard, these authors have also proved some asymptotic results for
the nodal (and critical point) distributions of complexified
restrictions of Dirichlet and Neumann eigenfunctions along strictly
convex, real-analytic interior curves $C \subset \Omega$. In this
case,  at least for an ergodic sequence of eigenfunctions, the
number of {\em complex} zeros of the holomorphic continuations of
the eigenfunction restrictions, $\phi_j|_{C},$ is $\sim c  \lambda.$
The same result is likely true when $C = \partial \Omega$  and this
would of course be consistent with the random result in Theorem
\ref{mainthm}.

In examples like the torus or sphere, where the spectrum of $\Omega$
is degenerate with high multiplicity, rather than summing
eigenfunctions belonging to different eigenspaces, it is natural to
consider the ensemble of {\em random eigenfunctions} attached to
{\em fixed} eigenspaces. A natural way to do so is to fix a basis
$B=\{\eta_{1}, \ldots \eta_{\mathcal{N}} \}$ of the eigenspace
$\mathcal{E}_{\lambda}$ and consider the random ensemble of
functions on $\Omega$ defined by
\begin{equation*}
\eta=a_{1}\eta_{1}+\ldots +a_{\mathcal{N}}\eta_{\mathcal{N}},
\end{equation*}
where $a_{i}$ are standard Gaussian i.i.d. Note that the probability
density of $\eta$ is independent of the choice of the basis $B$.

Berard \cite{Be} computed the expected length for the nodal line of
a random eigenfunction on the sphere to be $const\cdot \lambda$.
Rudnick and Wigman ~\cite{RW} and Wigman ~\cite{W} have studied the
variance of the length of the nodal line of random eigenfunctions
with $\lambda\rightarrow\infty$, for the torus and the sphere
respectively.

Recently, Granville and Wigman ~\cite{GW} have determined the
asymptotics of the variance of number of zeros of random
trigonometric polynomials of degree $\sim \lambda$ and moreover,
proved a central limit theorem for their distribution.   While there are
clearly similarities between the boundary traces of \eqref{eq:PhiaL
def} and \eqref{eq:PhiaS def} and random trigonometric polynomials
of degree $\sim \lambda$, it  would likely be difficult to prove a
central limit theorem for nodal distributions of random waves  on arbitrary smooth domains. However,   we do hope to
study the variance and higher moments in future work.

We thank Zeev Rudnick for many helpful discussions about random
zeros,  Steve Zelditch for helpful comments regarding pointwise
Weyl laws at the boundary and the anonymous referee for helpful comments and suggestions. The second author would also like to thank the CRM
analysis laboratory and its members for their support.

\section{A preliminary lemma on nodal intersections with $\partial\Omega$}
As above, we let $\Omega$ be a smooth bounded domain in ${\mathbb
R}^{2}$ and let $\varphi_j$ be the $L^2$ normalized eigenfunction of
the Laplacian with eigenvalue $\lambda_j^{2}$ satisfying the
Dirichlet boundary conditions. Without loss of generality, we assume
that the eigenfunctions are real-valued and we  let  $$q:[0,\ell] \rightarrow
\partial \Omega, \, \theta \mapsto q(\theta) = (q_1(\theta),
q_2(\theta))$$ be the arclength parametrization of the boundary
$\partial \Omega$ with $\ell:= \ell(\partial\Omega)$.

Let $v_j:[0,\ell] \rightarrow {\mathbb R}$ be the boundary trace of
$\varphi_{j}$, that is,
\begin{equation} \label{normal}
v_j(\theta) := \partial_{\nu} \varphi_j (q(\theta)).
\end{equation}
Here, $\nu = \nu(q)$ denotes the unit outward-pointing normal to the
boundary. Recall that the main object of our interest are the real
valued  random variables $\Phi_{a}^{L}(x;\lambda)$ and
$\Phi_{a}^{S}(x;\lambda)$ (see \eqref{eq:PhiaL def} and
\eqref{eq:PhiaS def}).

We would like to compute the leading asymptotics of
the expectation ${\mathcal Z}^{L,S} (\lambda) = \E
\intersect^{L,S}_{a}(\lambda)$ and we do this by  counting the zeros of the
boundary traces of $\Phi^{L}_{a}$ and $\Phi^{S}_{a}$ defined by
\begin{equation}
\label{eq:VaL def} V_{a}^{L}(\theta;\lambda) = \sum\limits_{
\lambda_j \in [0,\lambda] } a_j v_j(\theta),
\end{equation}
and
\begin{equation}
\label{eq:VaS def} V_{a}^{S}(\theta;\lambda) = \sum\limits_{
\lambda_j \in [\lambda,\lambda+1] } a_j v_j(\theta).
\end{equation}

The functions $V_{a}^{L,S}(\cdot;\lambda)$ are useful, since, as we
show in the next section, their zeros correspond to the
intersections of the nodal sets ${\mathcal N}_{a}^{L,S}(\lambda)$ with
the boundary $\partial \Omega$, for a {\em generic} choice of $a$.
In particular, we will show that the number
$\intersect^{L,S}_{a}(\lambda)$ of the intersections of the nodal line
${\mathcal N}_{a}^{L,S}(\lambda)$ with the boundary equals the number of
the zeros of $V_{a}^{L,S}(\theta;\lambda)$, almost surely (see Lemma
\ref{hopf}). This observation implies that
\begin{equation}
\label{eq:Z=ELLL} \mathcal{Z}^{L,S}(\lambda)
:=\E[\intersect^{L,S}_{a}(\lambda)]=\E[\tilde{\intersect}^{L,S}_{a}(\lambda)],
\end{equation}
where
\begin{equation*}
\tilde{\intersect}^{L,S}_{a}(\lambda) = \#\{\theta\in [0,\ell]:\:
V_{a}^{L,S} (\theta;\lambda) = 0 \}.
\end{equation*}

Our main interest here  is in  the asymptotic behaviour of
$\Z^{L,S}(\lambda)$, as $\lambda\rightarrow\infty$. In both cases we
reduce the problem of counting the nodal intersections to counting
the number of zeros of random functions, for which we use the
Kac-Rice formula. There are related formulas that have been used by
several authors in different settings: In a series of papers starting from
\cite{SZ} such formulas are used to study the distribution of zeros of random holomorphic
sections of vector bundles. A more classical reference for such
formulas is \cite{CL}.

\subsection{Open nodal lines and boundary critical points.}
In the Dirichlet case treated here, the key to computing
the $\lambda \rightarrow \infty$ asymptotics of $\Z( \lambda)$ is
the observation that {\em generically} the number of boundary
critical points of $\Phi_{a}(\theta;\lambda)$ equals twice the
number of open nodal lines. The following elementary result makes
this correspondence more precise.

We begin with the following  result for  functions vanishing on
$\partial \Omega$. It is an elementary result which is probably
known (see e.g. ~\cite{BGS}), but  since we could not find a direct
reference, we include the proof here.
\begin{lemma} \label{hopf}
Let $\Phi \in C^{\infty}(\Omega)$ satisfy $\Phi(q) = 0$ for all $q
\in \partial \Omega$. Assume that the critical points of $\Phi$  on
$\partial \Omega$ are simple. Then the number of open nodal lines of
$\Phi $  intersecting $\partial \Omega$ equals $\frac{1}{2}$ the
number of boundary critical points.
\end{lemma}
\begin{proof}
Let $\mathcal{N}_{\Phi}$ be the nodal set of $\Phi$. First, we assume that
$q$ is an intersection of $\N_{\Phi}$ with $\partial\Omega$. To
show that $
\partial_{\nu} \Phi(q) = 0,$ we flatten out the boundary and introduce local normal coordinates $(x_1,x_2) \in (-\epsilon,
\epsilon)^{2}$ with $x(q) = (0,0)$ and also,
$x_2(p) = 0$ for all $ p \in
\partial \Omega$ with $x_{2} >0$ in $\Omega$.

We locally have either
\begin{equation}
\label{eq:case x1=f(x2) nod} (f(x_2), x_2) \in \N_{\Phi}
\end{equation}
 for some $f \in C^{\infty}(-\epsilon, \epsilon)$ with $f(0) = 0$, or alternatively,
\begin{equation}
\label{eq:case x2=g(x1) nod} (x_{1}, g(x_{1}) ) \in \N_{\Phi},
\end{equation}
for a $C^{\infty}(-\epsilon, \epsilon)$ function $g$ with
$g(0)=g'(0) = 0$, where $g$ doesn't vanish identically on any
interval containing $0$ (the latter case is if the nodal line is
tangent to the boundary; only the former case is possible, see the
proof of the converse statement below). In fact, the simple zeros
assumption ensures that locally there is an equality rather than an
inclusion (see the proof of the converse statement).

In case \eqref{eq:case x1=f(x2) nod}, the Taylor expansion of
$\Phi(x_{1}, x_{2})$ around $(x_{1},0)$ gives
$$ \Phi(f(x_2), x_2) = \partial_{x_2} \Phi(f(x_2),0) x_{2} + {\mathcal O}(x_{2}^2) = 0,$$
since $\Phi (x_{1},0) = 0$, so that
$$ \partial_{x_2} \Phi(f(x_2),0)  = {\mathcal O}(x_{2}).$$
Since $f(0)=0$, it follows that $\partial_{x_2} \Phi(0,0) = 0$,
which finishes the proof in that case.

In case \eqref{eq:case x2=g(x1) nod}, similar reasoning gives
\begin{equation*}
\partial_{x_2} \Phi(x_1,0)  = {\mathcal O}(g(x_{1})),
\end{equation*}
provided that $g(x_1)$ does not vanish identically. Since $g(0)=0$,
it follows again that $\partial_{x_2} \Phi(0,0) = 0$.

To prove the converse, assume that $q \in \partial \Omega$ is a critical point
of $\Phi$ and show that $q$ is a nodal intersection point with the
boundary, $\partial \Omega$.  Using the same normal coordinates as
in the previous proof, we consider the  equation
\begin{equation} \label{implicit1}
\Phi(x_1,x_2) = 0,
\end{equation}
where, by assumption $ \partial_{x_2} \Phi(0,0) = 0.$ The fact that $\Phi(x_1,0) =0$ implies that
\begin{equation*}
\Phi(x_{1},x_{2}) = x_{2} \psi(x_{1},x_{2})
\end{equation*}
with $\psi(0,0)=0$ and $\partial_{x_{1}} \psi(0,0) \ne 0$. It
follows by the implicit function theorem that $\psi(x_{1},x_{2})=0$ is locally the graph
$x_{2}=g(x_{1})$.

\end{proof}

\section{A Kac-Rice formula for $\mathcal{Z}^{L,S}(\lambda)$}
\label{sec:int form exp intr}

\subsection{Some notation}
\label{sec:bi, cij def} Given a point $\theta\in [0, \ell]$ and the
spectral parameter $\lambda\in\R^{+}$, we introduce the vectors
$b_{1}^{L,S}(\theta;\lambda), b_{2}^{L,S}(\theta;\lambda) \in
\R^{N^{L,S}(\lambda)}$ where
\begin{equation}
\label{eq:b12L def} b_1^{L}(\theta;\lambda) = (v_{k}(\theta))_{
\lambda_k \in [0,\lambda]};\quad b_2^{L}(\theta;\lambda)
=(\partial_\theta v_{k}(\theta))_{ \lambda_k \in [0,\lambda]}
\end{equation}
for the long spectral range, and
\begin{equation}
b_1^{S}(\theta;\lambda) = (v_{k}(\theta))_{ \lambda_k \in
[\lambda,\,\lambda+1]}; \quad b_2^{S}(\theta;\lambda)=
(\partial_\theta v_{k}(\theta))_{ \lambda_k \in
[\lambda,\,\lambda+1]}
\end{equation}
for the short range. Here the dimension is given by \eqref{eq:spec
fnc long short def}. Note that for $\lambda$ sufficiently large, for every $\theta\in [0,
 \ell]$, the vectors $b_{1}^{L,S}(\theta;\lambda)$ and
$b_{2}^{L,S}(\theta;\lambda)$ are not collinear (see Corollary
\ref{cor:b1, b2 not colinear}). We will use this fact in the proof
of the main result of this section (Proposition \ref{mainprop}).

In addition to the vectors $b^{L,S}(\theta;\lambda)$ of boundary
traces of eigenfunctions it is also useful to define the
corresponding functions
$c_{ij}^{L,S} (\theta;\lambda):[0,\ell]\rightarrow \R$ for $i,j\in\{ 1,2 \}$,
where
\begin{equation*}
c_{ij}^{L}(\theta;\lambda) := \langle b_{i}^{L}(\theta;\lambda),\,
b_{j}^{L}(\theta;\lambda)\rangle;\quad c_{ij}^{S}(\theta;\lambda):=
\langle b_{i}^{S}(\theta;\lambda),\,
b_{j}^{S}(\theta;\lambda)\rangle.
\end{equation*}
For example, $c_{11}^{L,S}$ is just the squared length of
$b_{1}^{L,S}$ and $c_{12}^{L,S} = c_{21}^{L,S}$.

\subsection{A Kac-Rice formula}
In this section we give a Kac-Rice type formula for computing the
expected value of the number of zeros of a random combination  of
Dirichlet eigenfunctions. Results similar to the one here can be
found in ~\cite{CL} (see pg. 285).  Bleher, Shiffman and Zelditch
~\cite{BSZ1}, \cite{BSZ2} have proved a Kac-Rice    formula which applies in a more general situation.
However, for the convenience of the
reader, we give a direct elementary proof of Proposition \ref{mainprop} in the appendix.

The following Lemma \ref{lem:dbl zer prob 0} is used in the proof of Proposition \ref{mainprop}.
It implies that with probability $1$ all the zeros of
the boundary trace are {\em simple} and together with Lemma \ref{hopf} this implies the equality
\eqref{eq:Z=ELLL}, i.e. that the expected number of nodal intersection with the
boundary equals the expected number of zeros of the boundary trace.
The proof of Lemma \ref{lem:dbl zer prob 0} is also given in the appendix.

\begin{lemma}
\label{lem:dbl zer prob 0} For $\lambda$ sufficiently large, the set
$$ {\mathcal C} = \{a\in\R^{N^{L,S}(\lambda)}:\exists\theta\in
[0,\ell].\: V_{a}^{L,S}(\theta;\lambda) = \partial_{\theta}
{V_{a}^{L,S}}(\theta;\lambda)= 0\}$$ satisfies $\mu ({\mathcal
C})=0$ where $d\mu(a) := (2\pi)^{-N^{L,S}(\lambda)/2}
e^{-\|a\|^{2}/2} da.$ Moreover, $$codim\,{\mathcal C} \geq 1.$$
\end{lemma}

\begin{proposition} \label{mainprop}
The expected number of nodal intersections
with the boundary $\partial \Omega$ is given by
\begin{equation}
\label{kacrice} \Z ^{L,S}(\lambda)=
\frac{1}{\pi}\int\limits_{0}^{\ell}
\sqrt{\frac{c_{22}^{L,S}(\theta;\lambda)}{c_{11}^{L,S}(\theta;\lambda)}
-
\bigg(\frac{c_{12}^{L,S}(\theta;\lambda)}{c_{11}^{L,S}(\theta;\lambda)}\bigg)^2}
d\theta,
\end{equation}
with $c_{ij}^{L,S}(\theta;\lambda)$ defined in section \ref{sec:bi,
cij def}.

\end{proposition}

In the following sections \ref{longrange} and
\ref{shortrange}, we show that the leading term on the RHS of \eqref{kacrice} is
$\sim C_{\Omega}^{L,S} \lambda$ for some universal  constants
$C_{\Omega}^{L,S}>0.$

\section{Asymptotics for ${\mathcal Z}^{L}(\lambda).$} \label{longrange}

In this section, we compute the large-$\lambda$ asymptotics of the
pointwise sums $c_{11}^{L}(\theta;\lambda) =\sum_{\lambda_j \leq \lambda} |v_{j}(\theta)|^{2}$ and
$c_{22}^L(\theta;\lambda)= \sum_{\lambda_j
\leq \lambda} |\partial_{\theta} v_{j}(\theta)|^{2}.$ The leading asymptotics for the special case of  $c_{11}^{L}(\theta;\lambda)$ was computed by Ozawa \cite{Oz} and the asymptotics for the pointwise sum $c_{22}^{L}(\theta;\lambda)$  is closely-related, but we could not find it in the literature. For the benefit of the reader, we give the computation of both of the diagonal sums here. Alternatively, the leading asymptotics for  $c_{11}^{L}(\theta;\lambda)$ and $c_{22}^{L}(\theta;\lambda)$ follow from wave analysis in section \ref{shortrange}. However, since it is more elementary to  use the heat calculus for manifolds with boundary (and  also provides an independent verification of the asymptotics), this is the approach we take in this section.  The required bound for the mixed sum $c_{12}^{L}(\theta;\lambda)$ is more subtle and requires the full two-term asymptotics arising from the wave analysis in section \ref{shortrange} (see (\ref{defer})).

The relevant heat parametrix construction goes back to work of R. Seeley \cite{S} and here we briefly
recall the main results  that will be needed for our computations. The reader can find more detailed treatments in \cite{HZ} Appendix
12, \cite{S} and \cite{Oz}. Roughly speaking, the Seeley parametrix for the Dirichlet (or
Neumann) resolvent for an elliptic boundary value problem is
constructed as a sum of an interior parametrix and a
Poisson-kernel-type correction which compensates for the boundary
conditions. To describe it in more detail, it is useful to introduce
normal coordinates $(x,y) \in {\mathbb R}^{2}$ in a neighbourhood of
the boundary so that the Euclidean metric takes the form $dx^{2} +
h(x,y) dy^{2}$ and the domain is given in these local coordinates by
$x\geq 0$, the boundary corresponding to $x=0$. Consequently, since
$q:[0,\ell] \rightarrow \partial \Omega$ is the arclength
parametrization,
\begin{equation} \label{correspondence}
d\theta^{2} = h(0,y) dy^{2}.
\end{equation}

 We let $(x,y) \in U$
where $U$ is a sufficiently thin tubular neighbourhood of the
boundary, $\partial \Omega$.

 In local coordinates, the $N$-th order Seeley parametrix  for the Dirichlet
resolvent $R_{\Omega}(\mu):= (-\Delta_{\Omega} - \mu)^{-1}$ (here, $\mu = \lambda^2$) is of the form
 \begin{equation} \label{seeley 1}
R_{\Omega,N}(\mu) = \sum_{j=0}^{N} (2\pi)^{-2} \int_{\R^{2}} \int_{\R^{2}} e^{i(x-x')\xi} e^{i(y-y')\eta} c_{-2-j}(x,y,\xi,\eta,\mu) \, h^{-1/2}(x,y) \, d\xi d\eta \end{equation}
$$+ \sum_{j=0}^{N} (2\pi)^{-2}  \int_{\R^2} \int_{\R^2} e^{-ix' \xi} e^{i(y-y')\eta}
d_{-2-j}(x,y,\xi,\eta, \mu)  \, h^{-1/2}(x,y) \, d\xi d\eta.$$
In (\ref{seeley 1}) the factor $h^{-1/2}(x,y)$ is included in the integrand to cancel the coefficient in the volume measure $h^{1/2}(x,y) |dx dy|$ upon integration.
The first term on the RHS  in \eqref{seeley 1} is the interior
parametrix and the second one is the boundary-compensating
parametrix.  Here, we are mainly interested in the
leading coefficients, $c_{-2}, d_{-2}$. A direct computaton  in
local coordinates shows that
\begin{equation} \label{seeley2}
c_{-2}(x,y,\xi,\eta,\lambda) = (\xi^{2} + |\eta|^{2}_{h} - \mu)^{-1}
\end{equation}
where,
$$ |\eta|^{2}_{h}:= h^{-1}(x,y) \eta^{2}.$$
For the subsequent terms in the asymptotic series
$\sum_{j=0}^{N} c_{-2-j},$ $$\sup_{x,y \in U} | D_{x,y}^{\alpha}
c_{-2-j}(x,y,\xi,\eta; - i \mu)| \leq C_{\alpha} ( |\xi| |  + |\eta|
+ |\mu|^{1/2})^{-2-j}; \,\,\,j\geq 1, |\alpha| \geq 0.$$ Here, we
abuse notation somewhat and denote the holomorphic continuation of
$c_{-2-j}$  to the cone $\Sigma = \{ (\zeta,\mu) \in {\mathbb
C}^{2} \times {\mathbb C}; \Im \mu \leq (|\Re \mu| + |\Re \zeta|^{2}
- |\Im \zeta|^{2}) \}$ with $\zeta (\xi, \eta)$ also by $c_{-2-j}$. For the boundary-compensating terms one has
$$ d_{-2}(x,y,\xi,\eta,\mu) = - \frac{ e^{- x \sqrt{|\xi|^{2}_{h} -
\mu} } }{ \xi^{2} + |\eta|^{2}_{h} - \mu}, \,\,\,\, \Re
\sqrt{|\eta|^{2}_{h} - \mu}
>0, \,\,\, x \geq 0.$$   In this  case, the subsequent
terms in the sum $\sum_{j=0}^{N} d_{-2-j}$ satisfy estimates of the
form
$$ \sup_{x,y \in U} | x^{\alpha} D_{x}^{\beta} d_{-2-j}(x,y,\xi,\eta,-i\mu) | \leq C_{\alpha,\beta} (|\xi| + |\eta| + |\mu|^{1/2})^{-1-j-|\alpha|
+ |\beta|}$$
$$ \times \exp ( - \delta_{\alpha,\beta} y  (|\xi| + |\eta| + |\mu|^{1/2})  \, );  \,\,\,\, \delta_{\alpha,\beta} \geq
0.$$

To get a parametrix for the Dirichlet heat kernel of $\Omega,$ one writes the heat kernel as a contour integral
\begin{equation} \label{contour}
e^{-t \Delta_{\Omega}} = \frac{i}{2\pi} \int_{\Gamma} e^{-t\mu} (-\Delta_{\Omega}-
\mu)^{-1} d\mu. \end{equation} Here, $\Gamma = \{ - L + s e^{\pm
i\theta}, s \geq 0 \}$ with $L>0$ and $\theta\in (0,\frac{\pi}{2})$, is a
wedge enclosing the spectrum of $\Delta_{\Omega}$.

Substitution of the resolvent parametrix \eqref{seeley 1} in the
contour integral \eqref{contour} gives the following expression for
the heat kernel parametrix in $U \times U:$
\begin{equation} \label{seeley3}
H_{N}(t)(x,y;x',y')= \sum_{j=0}^{N} (2\pi)^{-2}\int \int e^{i(x-x')\xi} e^{i(y-y')\eta} \gamma_j(x,y,\xi,\eta,t)  \, h^{-1/2}(x,y) \, d\xi d\eta
\end{equation}
$$ +  \sum_{j=0}^{N} (2\pi)^{-2} \int \int e^{-ix' \xi} e^{i(y-y')\eta} \delta_{j}(x,y,\xi,\eta, t)  \, h^{-1/2}(x,y) \, d\xi d\eta.$$
It is well-known that $H_{N}(t)$ is a good approximation to $e^{-t\Delta_{\Omega}}$ in the sense that for $t \geq 0,$
\begin{equation} \label{derivativebounds}
|\partial_{z}^{\alpha} \partial_{w}^{\beta} (e^{-t\Delta_{\Omega}
}(z,w) - H_{N}(t)(z,w)) | \leq C_{\alpha,\beta} t^{ [ -(n+\alpha +
\beta) + N+1]/2} e^{-\delta |z-w|^{2}/t}; \, \, \delta
>0\end{equation} where $n=\dim\Omega=2$. By the Cauchy integral formula,
$$\gamma_{0}(x,y,\xi,\eta,t)= \frac{i}{2\pi} \int_{\Gamma} e^{-t \mu} c_{-2}(x,y,\xi,\eta,\mu) d\mu =  \frac{i}{2\pi} \int_{\Gamma} \frac{ e^{-t \mu}}{ \xi^{2} + |\eta|_{h}^{2} - \mu}  d\mu = e^{-t (\xi^{2} + |\eta|^{2}_{h})}.$$ Similarily,
$$\delta_{0}(x,y,\xi,\eta,t)= \frac{i}{2\pi} \int_{\Gamma} e^{-t \mu} d_{-2}(x,y,\xi,\eta,\mu) d\mu = - \frac{i}{2\pi} \int_{\Gamma}  e^{-t \mu} \frac{ e^{-  x \, \sqrt{|\eta|^{2}_{h} - \mu} } }{ \xi^{2} + |\eta|^{2}_{h} - \mu} d\mu$$
 $$=  - e^{-t (\xi^{2} + |\eta|^{2}_{h})} \times e^{-i x \xi}.$$
  So, subsititution of the formulas for $\gamma_0$ and $\delta_0$ and differentiation under the integral sign in \eqref{seeley3} gives

\begin{equation} \label{heatsum}
\sum_{j} e^{-t \lambda_j^{2}}  | v_j(\theta(y))|^{2} = \partial_{x}
\partial_{x'} H(t)(x,y;x',y')|_{x=x'=0,y=y'} + {\mathcal
O}(t^{-\frac{3}{2}}).
\end{equation}

$$=(2\pi)^{-2}  h^{-1//2}(0,y) \, \int \int e^{-t \xi^{2}} \,  \partial_x \partial_{x'}  \, [ \,  e^{-t  |\eta|_{h}^{2} } e^{-ix'\xi} \, ( e^{ix\xi} - e^{-ix\xi} ) \, ]|_{x=x'=0} \, d\xi d\eta +  {\mathcal O}(t^{-\frac{3}{2}})$$
$$=  (2\pi)^{-2} h^{-1/2}(0,y) \int \int e^{-t \xi^{2}} \,  (-i \xi ) \, \partial_x  \, [ \,  e^{-t  |\eta|_{h}^{2} } \, ( e^{ix\xi} - e^{-ix\xi} ) \, ]|_{x=0} \, d\xi d\eta +  {\mathcal O}(t^{-\frac{3}{2}})$$
In the last line, differentiation of $e^{-t |\eta|^{2}_{h}}$  gives an integrand that is odd in $\xi$ and so it integrates to zero.  Thus,
$$ \sum_{j} e^{-t \lambda_j^2} |\partial_{\nu} \phi_j(0,y)|^{2} = 2 (2\pi)^{-2}   h^{-1/2}(0,y) \, \int e^{-t \xi^2} \xi^{2} d\xi \times \int e^{-t |\eta|_{h}^2(0,y)}  d\eta
+ {\mathcal O}(t^{-\frac{3}{2}})$$
$$= a(y) t^{-2} + {\mathcal O}(t^{-\frac{3}{2}}),$$
where $a(y) = \frac{1}{4\pi}.$  By the  Karamata Tauberian theorem,
\begin{equation} \label{longsum1}
c_{11}^{L}(\theta;\lambda) = \sum_{j; \lambda_j^2 \leq \mu} |v_{j}(\theta)|^{2} \sim_{\lambda \rightarrow \infty}
c_{11}^{L}(\theta) \, \mu^{2} = c_{11}^{L}(\theta) \, \lambda^{4},
\end{equation} where,\begin{equation} \label{long1}
c_{11}^{L}(\theta) = (8\pi)^{-1}.
\end{equation}

 Similarily, the asymptotics for the second pointwise Weyl sum  $\sum_{\lambda_j^2
\leq \mu} |\partial_{\theta} v_j(\theta(y))|^{2}$ follows from the formula
\begin{equation} \label{longsum2+}
\sum_{\lambda_j^2 \leq \mu}
e^{-t \lambda_{j}^{2}}  |\partial_{\theta} v_{j}(\theta)|^{2} = (2\pi)^{-2} h^{-1/2}(0,y)  |\partial_{\theta} y|^{2} \cdot [ \partial_{y} \partial_{y'} \partial_x \partial_{x'}H(t)(x,y;x',y')] |_{x=x=0,y=y'} \end{equation}
$$ = 2 (2\pi)^{-2} h^{-1/2}(0,y) | \partial_{\theta} y|^{2} \int \int e^{-t \xi^2}  e^{- t h^{-1}(0,y) \eta^{2}} \xi^2 \eta^2  d\xi d\eta + {\mathcal O}(t^{-5/2}) \sim_{t \rightarrow 0^+} b(y)
 t^{-3},$$
 \vspace{1mm}
where, $ b(y) = \frac{1}{8\pi} h^{-1/2}(0,y) h(0,y)^{3/2} |\partial_{\theta}y|^{2}.$

There are Jacobian factors $\left| \partial_{\theta} y
\right| >0$ occuring in (\ref{longsum2+}) because of the change of
parametrization of the boundary given by $y \mapsto \theta(y)$ and
we also use that $\partial_\theta x=0.$ Differentiation under the
integral sign of the heat parametrix is justified since all
exponential sums of derivatives of boundary traces of eigenfunctions
are absolutely convergent in view of \eqref{derivativebounds} and
derivatives of the remainder term are controlled. Just as in
(\ref{longsum1}), the Tauberian theorem gives
\begin{equation} \label{longsum2}
c_{22}^{L}(\theta;\lambda) = \sum_{\lambda_j \leq \lambda} |\partial_{\theta}
v_j(\theta)|^{2} \sim_{\lambda \rightarrow \infty} c_{22}^{L}(\theta)
\lambda^6,
\end{equation}
where,
\begin{equation} \label{long2}
c_{22}^{L}(\theta) = \frac{1}{48 \pi}  [h(0,y(\theta))] \,
|\partial_{\theta} y|^{2} = (48\pi)^{-1}.
\end{equation} In (\ref{long2}) we have used that
$\theta(y) = \int_{0}^y [h(0,s)]^{1/2} ds$ (see
(\ref{correspondence})) and so, $|\partial_{\theta}y| =
[h(0,y(\theta))]^{-1/2}.$

Finally, we turn to  the mixed sum $c_{12}^{L}(\theta;\lambda).$ By the same heat analysis as for $c_{22}^{L}(\theta;\lambda),$ one shows that
$$ \sum_{\lambda_j \leq \lambda} |\partial_{\theta}v_j(\theta) + v_{j}(\theta)|^{2} \sim_{\lambda \rightarrow \infty} (48\pi)^{-1} \lambda^{6} $$ and so, by writing $2 c^{L}_{12}(\theta;\lambda) = \sum_{\lambda_j \leq \lambda} |(\partial_{\theta} +1)v_j(\theta)|^{2} - \sum_{\lambda_j \leq \lambda} |\partial_{\theta}v_{j}(\theta)|^{2} - \sum_{\lambda_j \leq \lambda} |v_{j}(\theta)|^{2}$ it follows that $c_{12}^{L}(\theta;\lambda) = o(\lambda^{6})$ uniformly for $\theta \in [0,\ell].$ Unfortunately, this bound is not sufficient to prove our theorem. Indeed, we will need the stronger bound
\begin{equation} \label{long3}
c_{12}^{L}(\theta;\lambda) = o(\lambda^{5}).\end{equation}
This estimate follows from analysis of the Dirichlet wave kernel and we defer the proof to section \ref{shortrange} (see (\ref{defer})).

\section{Asymptotics for ${\mathcal Z}^{S}(\lambda)$} \label{shortrange}
\subsection{Computation of $c_{11}^{S}(\theta;\lambda)$}
The heat analysis in the previous section is not sufficient to deal with the
short-range case, even for the diagonal sums. An alternative approach involves the Dirichlet wave operator. However, unlike the boundaryless case, one does not have an explicit wave parametrix, even for small time. This is the main complication in the case of manifolds
with boundary. Indeed, the small-time behaviour of the Dirichlet wave
operator kernel $E_{\Omega}(x,y;t); \, x,y \in \Omega$ is quite
complicated and in particular, has no conormal expansion at $t=0$.
Nevertheless, it turns out that the restriction to the boundary diagonal {\em does} have such an expansion (see
\eqref{expansion}). Using the fact that the singularity of
$E_{\Omega}^{\flat}(q,q;t); \, q \in \partial \Omega$ at $t=0$ is
isolated, Zelditch \cite{Z1} computed the asymptotics of $c_{11}^{S}(\theta
;\lambda).$ The idea is to use boundary wave front
calculus and the conormal expansion of Ivrii \cite{I} and Melrose
\cite{M1} together with the Hadamard variational formulas for the
Dirichlet (or Neumann) eigenvalues to compute the coefficients in
the conormal expansion of the boundary trace of the interior wave
group restricted to the boundary diagonal. The asymptotic expansion of
$c_{22}^{S}(\theta;\lambda)$ and the bound for $c_{12}^{S}(\theta;\lambda)$ in the next subsection are apparently new.
Since the computations for $c_{22}^{S}(\theta;\lambda)$ and $c_{12}^{S}(\theta;\lambda)$ are related to the one for
$c_{11}^{S}(\theta;\lambda)$ (and use it), we give the argument for each of spectral sums.


Let
$$E^{\Omega}(t) = \cos t \sqrt{\Delta_{\Omega}} $$
be the Dirichlet wave operator where $\Delta_{\Omega}$ denotes the
Dirichlet Laplacian on $\Omega$. To fix notation, we denote the
boundary restriction operator by $\gamma_{\partial \Omega}$, so that
for $u \in {\mathcal D}'(\Omega)$ with $WF(u) \cap N^* \partial
\Omega = \emptyset, \, \gamma_{\partial \Omega}u$ is the
uniquely-defined restriction (ie. pull-back of $u$) with
$\gamma_{\partial \Omega}u \in {\mathcal D}'(\partial \Omega)$ and
$\gamma_{\partial \Omega}u = u|_{\partial \Omega}$ when $u \in
C^{\infty}(\Omega).$

The diagonal kernel $t^{n}E_{\Omega}(t,x,x) \notin C^{\infty} ((0, \epsilon]
\times  \Omega \times \Omega )$ for any $n>0$ and $\epsilon >0$ when $x\in \Omega$.
The singularity at $t=0$ is not even isolated due to reflecting
loops based at $ x \in \Omega$ with period $t = 2 d(x,\partial
\Omega)$. However, when $x= q \in
\partial \Omega$, the singularity at $t=0$ is isolated. The reason
for this can be described as follows: let
$$ \Phi_{t}: T^*\Omega \rightarrow T^*\Omega$$ be the time-$t$ broken generalized bicharacteristic flow in $\Omega$ \cite{MSj}. For $(q,\eta) \in B^{*}
\partial \Omega$ we let $\xi(q,\eta) \in S^{*}_{+}(\Omega)$ be the
unique inward pointing unit vector that projects tangentially onto
$\eta$. We denote the tangential projection of $v \in
S^{*}_{+}(\Omega)$ here by $v^{T} \in B^{*}(\partial \Omega)$ and
$\gamma_{\partial \Omega}: \phi \mapsto \partial_{\nu}
\phi|_{\partial \Omega}$ the boundary trace operation.

Propagation of singularities in $\Omega$ \cite{MSj, M1} gives
\begin{equation} \label{wf1}
WF( E_{\Omega}(x,y,t)) = \{ (t,\tau,x,\xi,y,\eta); \tau = |\xi|_{g},
\Phi_{t}(x,\xi) = (y,\eta) \}.\end{equation}  For fixed $t \in {\mathbb R}$ we define the distribution $E_{\Omega,t}$  on $\Omega \times \Omega$ given by
$$ E_{\Omega,t}(x,y)= E_{\Omega}(x,y,t).$$
 Let
$N^{*}\Delta_{\partial \Omega} = \{ (q,\xi,q,-\xi) ; q \in \partial
\Omega \}$ be the conormal space of the diagonal $\Delta_{\partial
\Omega} \subset \partial \Omega \times \partial \Omega $. In view of
\eqref{wf1}, when $t \in (0,\epsilon) $ with $\epsilon >0$
sufficiently small,
$$ WF( E_{\Omega,t}) \cap N^{*}( \Delta_{\partial \Omega}) = \emptyset.$$
This just says that for $|t|$ sufficiently small, there are no
non-trivial, periodic broken bicharacteristics passing through $q \in \partial
\Omega$ with period $|t|$. Then by wave-front calculus, $E^{\flat}:= \gamma_{\Delta \partial \Omega} E_{\Omega}$ is well-defined for
$\epsilon >0$ small and for $t \in (0,\epsilon),$
\begin{equation} \label{wf2}
WF( E^{\flat}(q,q,t)) = \{ (t,\tau,q,\eta,q,\eta') \in
T^{*}(-\epsilon,\epsilon) \times B^*\partial \Omega \times
B^*\partial \Omega  ; \tau = |\eta|_{g}, \end{equation}
$$  [\Phi_{t}(q,\xi(q,\eta) )]^{T}
= (q' ,\eta') \}. $$
In \eqref{wf2},  $\xi(q,\eta) \in S^*_{in,q}
\Omega$ is the inward-pointwing unit co-vector at $q \in \partial
\Omega$ with $[\xi(q,\eta)]^{T} = \eta$ where $v^{T}$ is the
tangential projection at $q \in \partial \Omega$.

It follows that the singularity at $t=0$ is classical conormal
\cite{I,M1}, so that for $\epsilon >0$ small as $t \rightarrow 0^+,$
one has $t^{n} E^{\flat}(t,q,q) \in C^{\infty}$ where $n = \dim
\Omega$ and so,
\begin{equation} \label{expansion}
t^{n} E^{\flat}(t,q,q) \sim \sum_{k=0}^{\infty} a_k(q) t^{k}
\end{equation}
with $a_{j} \in C^{\infty}(\partial \Omega)$.  Under the non-recurrence assumption on the billiard flow, the expansion in (\ref{expansion}) holds, modulo $C^{\infty}(  (0,\epsilon^{-1}) \times \partial \Omega \times \partial \Omega),$  for all $t \in (0,\epsilon^{-1})$ where $\epsilon >0$ is fixed arbitrarily small.

Since  the method of proof
in \cite{I,M1} is non-constructive, one needs an additional argument
to compute the actual coefficient functions  $a_{j}.$ One way to do this is to integrate both sides in \eqref{expansion} against
a test function $\psi \in C^{\infty}(\partial \Omega)$. This is
allowed because $t^{n} E^{\flat} \in C^{\infty}$ and the result is that
\begin{equation} \label{coeff1}
  t^{n} \sum_j e^{i\lambda_j t}  \int_{\partial \Omega} |\partial_{\nu} \phi_{j}(q)|^{2} \psi(q) d\sigma(q)  \sim_{t \rightarrow 0^+} \sum_{k=0}^{\infty} t^{k} \int_{\partial \Omega}  a_{k}(q) \psi(q) d\sigma(q).
\end{equation}
Consider a normal variation of the domain $\Omega$ with variation
vector field, $\psi \nu$.  By taking the variation $\delta_{\psi}$  of both sides in (\ref{coeff1}) and using the Hadamard variation formula for
the eigenvalues:
$$ \delta_{\psi} \lambda_j^{2} = 2 \lambda_j \delta_{\psi} \lambda_j = \int_{\partial \Omega} \psi(q) |\partial_{\nu} \phi_{j}(q)|^{2} d\sigma(q),$$ it follows
that
$$ \frac{1}{i} \frac{\partial}{\partial t}  \left( \frac{2}{it} \sum_{j} \delta_{\psi}(e^{it\lambda_j}) \right) = \frac{1}{it} \sum_j 2 \lambda_j \delta_{\psi} (e^{it \lambda_j}) = \sum_j (\delta_{\psi} \lambda_j^2) e^{i\lambda_j t}  $$
\begin{equation} \label{coeff2}
= \sum_{j}  e^{i\lambda_j t}  \int_{\partial \Omega} \psi(q) |\partial_{\nu} \phi_j(q)|^{2} d\sigma(q).
\end{equation}
 By a well-known asymptotic expansion for the wave trace \cite{I,M1}
\begin{equation}\label{coeff3a}
\sum_{j} e^{i\lambda_j t} \sim_{t \rightarrow 0^+} (2\pi)^{-2} vol
(\Omega) (t + i0)^{-2} + C' vol (\partial \Omega) (t +
i0)^{-1} + ...
\end{equation}
where the dots denote lower-order terms.  Substitution of (\ref{coeff3a}) in (\ref{coeff2}) gives
\begin{equation} \label{coeff3b}
\frac{1}{i} \frac{\partial}{\partial t} \left( \frac{2}{it} \sum_{j} \delta_{\psi}( e^{i\lambda_j t})  \right)  \sim_{t \rightarrow 0^+} 2 (2+1) (2\pi)^{-2} \delta_{\psi} vol
(\Omega) (t + i0)^{-4} + C' \delta_{\psi} vol (\partial \Omega) (t +
i0)^{-5} + ...
\end{equation}
Equating coefficiens in
\eqref{coeff1}-\eqref{coeff3a} implies that
\begin{equation} \label{upshot1}
\sum_{j} e^{i\lambda_j t} \int_{\partial \Omega} \psi(q)
|\partial_{\nu} \phi_{j}(q)|^{2} d\sigma(q) \sim_{t \rightarrow 0^+}
6 (2\pi)^{-2} \delta_{\psi} vol(\Omega) (t+i0)^{-4} + C' \delta_{\psi} vol
(\partial \Omega) (t+i0)^{-3} + ...,
\end{equation}
where, a direct computation gives
$$ \delta_{\psi}  vol(\Omega) = \int_{\partial \Omega} \psi(q) d\sigma(q). $$
 Under the assumption that there are measure zero loops at the boundary  it follows from (\ref{coeff1}) by a standard Tauberian argument (see for example \cite{Z1} section 5.2) that 
\begin{equation} \label{shortsum1}
c_{11}^{S}(\theta,\lambda)= \sum_{\lambda_j \in [ \lambda, \lambda +
1] } |  v_j(\theta)|^{2}  =  c_{11}^{S}(\theta)\cdot
\lambda^{3} + o(\lambda^{2}),\end{equation}
and since the test function $\psi \in C^{\infty}(\partial \Omega)$ is arbitrary, it follows from (\ref{upshot1}) that
\begin{equation} \label{short1}
c_{11}^{S}(\theta) =  (2\pi)^{-1}. \end{equation}
\subsection{Computation of $c_{22}^{S}(\theta;\lambda)$} The computation of $c_{22}^{S}$ uses the asymptotics for  $c_{11}^{S}$ above together with some additional integration by parts and applications of the local Weyl law \cite{HZ} for boundary traces of eigenfunctions.
Let $\partial_{\theta_q}:C^{\infty} (\partial \Omega) \rightarrow C^{\infty}(\partial \Omega)$ be the tangential derivative given by $\partial_{\theta_q}:= dq (\partial_{\theta}).$
Since $WF(     \partial_{\theta_q} ) \subset \{ (q, \xi; q,
- \xi); \, (q,\xi) \in T^* \partial \Omega \},$ by wave front calculus,
\begin{equation} \label{wf4}
WF(  [\partial_{\theta_q} E_{\Omega} \partial_{\theta_q'} ]^{\flat}
(q,q',t) )  = WF(  [\partial_{\theta_q} E_{\Omega}^{\flat}
\partial_{\theta_q'} ](q,q',t) )  \subset WF(
E_{\Omega}^{\flat}(q,q',t)).
\end{equation}

Then by the same argument as in the previous section, it follows
that $ \partial_{\theta_q}  E_{\Omega}^{\flat} \partial_{\theta_q'}$
has a unique restriction to the boundary diagonal which we denote by
$[  \partial_{\theta_q}  E_{\Omega}^{\flat} \partial_{\theta_q'} ]
(q,q,t):=  \gamma_{\Delta_{\partial \Omega}} [\partial_{\theta_q}
E_{\Omega}^{\flat} \partial_{\theta_q'}].$  Since wave fronts
restrict, under the non-recurrence assumption,  for $t \in (0,\epsilon^{-1})$ with
$\epsilon >0$ arbitrarily small, modulo $C^{\infty}((0,\epsilon^{-1}) \times \partial \Omega \times \partial \Omega),$ one has the following conormal expansion
\begin{equation} \label{conormal2}
[\partial_{\theta_q}  E_{\Omega}^{\flat} \partial_{\theta_q'} ](q,q,t) \sim_{t \rightarrow 0^+} \sum_{k=0}^{\infty} b_{k}(q)  t^{-n-2+k}
\end{equation}
where $b_k \in C^{\infty}(\partial \Omega)$ and  the corresponding two-term asymptotic formula
$$ \sum_{\lambda_j \leq \lambda} |\partial_{\theta} v_j(\theta)|^{2} = c_{22}^{S}(\theta) \lambda^{5} + o(\lambda^{4}).$$
 The remainder of this subsection is devoted to computing the  leading coefficient $c_{22}^{S}(\theta).$

From now on, we  put $\hbar_j = \frac{1}{\lambda_j}; j=1,2,3,...$ and use
semiclassical pseudodifferential calculus on $\partial \Omega.$
Let $A_{\hbar_j}=Op_{\hbar_j}(a)$ with $a \in S^{0}_{cl}(T^*
\partial \Omega)$ and $A_{\hbar_j}^{q}:= q^{-1} \circ
Op_{\hbar_j}(a) \circ q $ be the corresponding  semiclassical
pseudodifferential operator on the parametrizing circle ${\mathbb
R}/\ell \mathbb{Z} $. It follows that in the Dirichlet case
considered here \cite{HZ},
\begin{equation} \label{Weyl1}
\frac{1}{N_{0}(\lambda)} \sum_{\lambda_j \leq \lambda} \langle
\hbar_{j}^{2}  A^{q}_{\hbar_j} v_{j}, v_{j} \rangle =
\frac{2}{\pi vol \Omega} \int_{B^{*}\partial \Omega} a(y,\eta) \sqrt{1-|\eta|_{h}^{2}} \, dy d\eta  + o(1),
\end{equation}
where, $N_{0}(\lambda) := (2\pi)^{-2}  (vol B^*\Omega) \, \lambda^{2}.$
Since  we are assuming here that the measure of periodic broken bicharacteristics is
zero, using the conormal expansion in \eqref{conormal2} one gets that
for some constant $C_{a}$ (yet to be determined),
\begin{equation} \label{Weyl2}
\frac{1}{N_{0}(\lambda)} \sum_{\lambda_j \leq \lambda} \hbar_j^2 \langle
A^{q}_{\hbar_j} v_{j}, v_{j} \rangle = C_{a}+ o(\lambda^{-1}).
\end{equation} But then from the weak law in \eqref{Weyl1}  it
follows that  $C_{a} = \int_{B^*\partial \Omega} a(y,\eta)
d\sigma(y,\eta)$ where $d\sigma(y,\eta):= 2 (\pi \, \text{vol} \Omega)^{-1}  \sqrt{1-|\eta|_{h}^{2}} \, dy d\eta.$ Taking sums with $\lambda_j \leq \lambda + 1$ and
$\lambda_j \leq \lambda$ in \eqref{Weyl2} and subtracting them gives
\begin{equation} \label{Weyl3}
\sum_{\lambda_j \in [\lambda,\lambda +1]} \hbar_j^2 \langle A^{q}_{\hbar_j} v_j,
v_j \rangle = \int_{B^{*} \partial \Omega} a(y,\eta) d\sigma(y,\eta)
[ N_{0}(\lambda +1) - N_{0}(\lambda)]  + o_{A}(\lambda).
\end{equation}
Now rescale and write $\lambda_j = \hbar_j^{-1} $ and let $\psi \in
C^{\infty}({\mathbb
R}/\ell \mathbb{Z} )$. In analogy with the previous
section, we want to compute
\begin{equation} \label{coeff3}
\sum_{\hbar_j^{-1} \in [\lambda, \lambda + 1]} \int_{0}^{\ell} \psi
(\theta) |\partial_{\theta} v_j(\theta)|^{2} d\theta
=\sum_{\hbar_j^{-1} \in [\lambda, \lambda + 1]} \hbar_j^{-2} \langle
\psi (\hbar_j D_{\theta})^{2} v_j, v_j \rangle
\end{equation}
$$+  \sum_{\hbar_j^{-1} \in [\lambda, \lambda + 1]}  \hbar_j^{-2} \langle (\hbar_j D_{\theta} \psi) v_{j}, \hbar_j D_{\theta} v_j
\rangle$$
$$= \sum_{\hbar_j^{-1} \in [\lambda, \lambda + 1]} \hbar_j^{-2} \langle \psi (\hbar_j D_{\theta})^{2} v_j, v_j \rangle + \frac{1}{2} \sum_{\hbar_j^{-1}  \in [\lambda,  \lambda + 1]}  \langle \partial^{2}_{\theta}\psi,  |v_{j}|^{2}
\rangle$$

$$=  \sum_{\hbar_j^{-1} \in [\lambda, \lambda + 1]} \hbar_j^{-2} \langle \psi (\hbar_j D_{\theta})^{2} v_j, v_j \rangle + {\mathcal O}(\lambda^{3}).$$
The second last line in \eqref{coeff3} follows by integration by
parts and the last line by the result in the previous section for
$c_{11}^{S}(\theta;\lambda)$ (here, we also use that $\hbar_j^{-1} \in [\lambda, \lambda +1]$).
 Since $$WF'_{\hbar_j}( F(\hbar_j^{-1})) \subset  B^{*}_{1+\delta}(\partial \Omega) \times B^{*}_{1+\delta}(\partial \Omega),$$  from the boundary jump
equation (\ref{jump}) it follows that for any $\delta >0,$
$$ WF_{\hbar_j}(v_{j}) \subset (dq^t)^{-1} \, B^{*}_{1+ \delta}(\partial \Omega).$$
Let $\chi \in C^{\infty}_{0}( (dq^t)^{-1} B^{*}_{1+2\delta} \partial \Omega)$
with $\chi(\theta,\xi) =1$ for $(\theta,\xi) \in (dq^t)^{-1} B^{*}_{1+ \delta} \partial
\Omega.$ Then, from the last line in \eqref{coeff3}, $$\sum_{\hbar_j^{-1} \in [\lambda, \lambda + 1]} \hbar_j^{-2} \langle \psi (\hbar_j D_{\theta})^{2} v_j, v_j \rangle + {\mathcal O}(\lambda^{3}) $$
 $$ = \sum_{\hbar_j^{-1} \in [\lambda, \lambda + 1]}  \hbar_j^{-2} \langle
Op_{\hbar_j}(\chi) \psi (\hbar_j D_{\theta})^{2} Op_{\hbar_j}(\chi)
v_j, v_j \rangle + {\mathcal O}(\lambda^{3})$$
$$= \lambda^{4}( 1 + {\mathcal O}(\lambda^{-1})) \, \sum_{\hbar_j^{-1} \in [\lambda, \lambda + 1]}  \hbar_j^{2} \langle
Op_{\hbar_j}(\chi) \psi (\hbar_j D_{\theta})^{2} Op_{\hbar_j}(\chi)
v_j, v_j \rangle + {\mathcal O}(\lambda^{3})$$

Clearly,
$\psi Op_{\hbar_j}(\chi) (\hbar_j D_{\theta})^{2} Op_{\hbar_j}(\chi)
\in Op_{\hbar_j}(S^{0}_{cl})$ has principal symbol $$\sigma [\,   \psi \, Op_{\hbar_j}(\chi)
 |\hbar_j D_{\theta}|^{2} Op_{\hbar_j}(\chi) \, ]  (\theta,\xi) \,  =  \,  \psi(\theta)  \, |\xi|^{2},$$
for  $(\theta,\xi) \in  (dq^{t})^{-1} B^{*}\partial
\Omega.$  But then, since $\hbar_j^{-1} \in [\lambda, \lambda +1],$ it follows from
\eqref{Weyl3} that  one has the two-term asymptotic formula $c_{22}^{S}(\theta;\lambda) = c_{22}^{S}(\theta) \lambda^{5} + o(\lambda^{4}),$ where,
\begin{equation} \label{shortsum2}
c_{22}^{S}(\theta) = \lim_{\lambda \rightarrow \infty} \lambda^{-5} \frac{2  }{\pi \,   \text{vol}  \Omega}
\left(  \int_{|\xi| \leq 1} |\xi|^{2} \,
 \sqrt{ 1-|\xi|^{2}  } \, d\xi \right)  \,  \lambda^{4} [ N_0(\lambda +1) - N_0(\lambda)],  \end{equation} and  since $N_0(\lambda +1) -N_0(\lambda) = 2(2\pi)^{-2}  (vol B^*\Omega ) \, \lambda,$ it follows that
\begin{equation} \label{short2}
c_{22}^{S}(\theta) =  2 \times \frac{\pi}{8}   \times \frac{2}{\pi \text{vol} \, \Omega} \times (2\pi)^{-2} \text{vol}B^{*} \Omega  = (8 \pi)^{-1}.
 \end{equation}

\subsection{ Bound for $c_{12}^{S}(\theta;\lambda)$}
By the same argument as for $c_{22}^{S}(\theta;\lambda)$ one gets the two-term asymptotic formula,
\begin{equation} \label{newsum}
\sum_{\lambda_j \in [\lambda,\lambda +1]} |\partial_{\theta}v_j(\theta) + v_j(\theta)|^{2} = (8\pi)^{-1} \lambda^{5} + o(\lambda^{4}), \end{equation}
uniformly for $\theta \in [0,\ell].$ Then,
\begin{equation} \label{completesquare}
2c_{12}^{S}(\theta;\lambda) = \sum_{\lambda_j \in [\lambda,\lambda +1]} |\partial_{\theta}v_j(\theta) + v_j(\theta)|^{2} - \sum_{\lambda_j \in [\lambda,\lambda +1]} |\partial_{\theta}v_j(\theta)|^{2} - \sum_{\lambda_j \in[\lambda,\lambda+1]} | v_j(\theta)|^{2} = o(\lambda^4), \end{equation}
 where,  the final bound in (\ref{completesquare}) follows from the two-term asymptotic formulas in (\ref{short1}), (\ref{short2}) and (\ref{completesquare}).
The upshot is that
\begin{equation} \label{shortsum3}
c_{12}^{S}(\theta;\lambda)  = o(\lambda^4),\end{equation}
and so, by the triangle inequality, \begin{equation} \label{defer}
|c_{12}^{L}(\theta;\lambda)| \leq C_{\Omega} \lambda \, |c_{12}^{S}(\theta;\lambda)| = o(\lambda^5). \end{equation}

\section{Concluding the proof of theorem \ref{mainthm}}

\begin{corollary}
\label{cor:b1, b2 not colinear} There exists a constant $\lambda_0
>0$ such that for $\lambda \geq \lambda_0$ the vectors
$b_1^{L,S}(\theta;\lambda)$ and $b_2^{L,S}(\theta;\lambda)$ are
linearly independent for all $\theta \in [0,\ell].$
\end{corollary}
\begin{proof}
By Cauchy-Schwarz, the vectors $b_1^{L,S}(\theta;\lambda) =
(v_1(\theta),....,v_\lambda(\theta))$ and $b_1^{L,S}(\theta;\lambda)
= (\partial_{\theta} v_1(\theta),...., \partial_{\theta}
v_\lambda(\theta)) $  are collinear for some $\theta \in [0,\ell]$
if and only if
$$ c_{11}^{L,S}(\theta;\lambda) \cdot c_{22}^{L,S}(\theta;\lambda) - |c_{12}^{L,S}(\theta;\lambda)|^{2} = 0.$$
But by our pointwise Weyl sum computations in \eqref{long1},
\eqref{long2} and \eqref{long3},
$$  c_{11}^{L}(\theta;\lambda) \cdot c_{22}^{L}(\theta;\lambda) - |c_{12}^{L}(\theta;\lambda)|^{2}  \sim_{\lambda \rightarrow \infty} C_{\Omega}(\theta) \lambda^{10};  \,\,C_{\Omega}(\theta) >0,$$
and similarily, from \eqref{shortsum1}, \eqref{shortsum2} and
\eqref{shortsum3},
$$  c_{11}^{S}(\theta;\lambda) \cdot c_{22}^{S}(\theta;\lambda) - |c_{12}^{S}(\theta;\lambda)|^{2} \sim_{\lambda \rightarrow \infty} \tilde{C}_{\Omega}(\theta) \lambda^{8}; \,\, \tilde{C}_{\Omega}(\theta) >0.$$
So, for $\lambda>0$ large, neither expression vanishes.
 \end{proof}

\begin{remark}
\label{rem:dt, d2t ind} We note that by a similar analysis to the
one given in sections \ref{longrange} and \ref{shortrange} and
Corollary \ref{cor:b1, b2 not colinear} above, one easily shows that
the vectors $(\partial_{\theta} v_{j}(\theta))_{\lambda_j \in
[\lambda, \lambda +1]}$ and $(\partial_{\theta}^{2}
v_{j}(\theta))_{\lambda_j \in [\lambda, \lambda +1]}$ are linearly
independent for each $\theta \in [0,\ell]$. The analogous result for
the long-range case is also valid.
\end{remark}

\begin{proof}[Proof of theorem \ref{mainthm}]
Combining the asymptotic formulas for
$c_{ij}^{L,S}(\theta;\lambda),\,  i,j=1,2$  in \eqref{long1},
\eqref{long2} and \eqref{long3} and \eqref{short1},
\eqref{short2} and \eqref{shortsum3} gives
$$ \left( \frac{c_{12}^{L,S}(\theta;\lambda)}{c_{11}^{L,S}(\theta;\lambda)} \right)^{2}  =  o(\lambda)$$
uniformly for $\theta \in [0,\ell]$. On the other hand,
$$  \frac{c_{22}^{L,S}(\theta;\lambda)}{c_{11}^{L,S}(\theta;\lambda)} \sim_{\lambda \rightarrow \infty} \frac{c_{22}^{L,S}(\theta) }{c_{11}^{L,S}(\theta)}
\lambda^{2}$$ where,
$$ \frac{c_{22}^{L}(\theta) }{c_{11}^{L}(\theta)} =\frac{1}{6} \,\,\,\,\,\text{and} \,\,\,\,\, \frac{c_{22}^{S}(\theta) }{c_{11}^{S}(\theta)} = \frac{1}{4}.$$

So, substitution into the Kac-Rice formula in Proposition
\ref{mainprop}
$$ {\mathcal Z}^{L,S}(\lambda) = \frac{1}{\pi}  \left( \int_{0}^{\ell} \left|  \frac{c_{22}^{L,S}(\theta) }{c_{11}^{L,S}(\theta) }  \right|^{\frac{1}{2}}  \,  d\theta \right)  \, \lambda + o(\lambda),
$$
implies that ${\mathcal Z}^{S}(\lambda) = \frac{\ell(\partial
\Omega)}{2\pi} \lambda + o(\lambda)$
and ${\mathcal Z}^{L}(\lambda) = \frac{\ell(\partial
\Omega)}{\sqrt{6}\pi} \lambda + o(\lambda).$

\end{proof}

\appendix

\section{Proof of the Kac-Rice formula}

\begin{proof}[Proof of Lemma \ref{lem:dbl zer prob 0}]
The proof in the long and short range cases is the same, so without
loss of generality we treat the long range case here.

In the following, we fix $\lambda > 0$ and put $N=N^{L}(\lambda).$
Consider the map
\begin{equation*}
\Psi_{\lambda} : \R^{N}\times [0,\ell]\rightarrow \R^2
\end{equation*}
\begin{equation*}
(a,\theta)\mapsto \big( V_{a}^{L}(\theta;\lambda), \partial_{\theta} V^{L}_{a}
(\theta;\lambda)\big).
\end{equation*}
Clearly, $${\mathcal C} = \pi_{1}( \Psi_{\lambda}^{-1}(0,0)),$$
where $\pi_{1}$ is the projection onto $\R^{N}$. We claim that
$\Psi_{\lambda}$ is a submersion. Given this, one has
\begin{equation*}
\dim \Psi_{\lambda}^{-1}(0,0) = N-1,
\end{equation*}
and so, $\dim {\mathcal C} \le N-1$,
which proves the Lemma.

To see that $\Psi_{\lambda}$ is a submersion, we note that its
$(N+1)\times 2$ differential matrix is given by
\begin{equation*}
d\Psi_{\lambda} (a,\theta) = \left(\begin{matrix} b_{1} ^{L}(\theta;
\lambda) &b_{2}^{L}(\theta;\lambda) \\ * &*
\end{matrix}\right),
\end{equation*}
where $b_{1}^{L}(\theta;\lambda)$ and $b_{2}^{L}(\theta;\lambda)$
are the vectors introduced in section \ref{sec:bi, cij def}. The
matrix $d\Psi_{\lambda}$ is of full rank (i.e. rank $2$) for
$\lambda$ sufficiently large, by Corollary \ref{cor:b1, b2 not
colinear} and so $\Psi_{\lambda}$ is a submersion.

\end{proof}

\begin{definition}
Let $\chi_{[-1,1]}$ be the characteristic function of the interval $[-1,1]$ and fix $\epsilon >0$.
We introduce the random variables
\begin{equation}
\label{eq:LLLeps def} \numzeros_{a,\epsilon}^{L,S}(\lambda)
:=\frac{1}{2\epsilon}\int\limits_{0}^{\ell} \chi_{[-1,1]}
\bigg(\frac{V_{a}^{L,S}(\theta;\lambda)}{\epsilon}\bigg)
\big|\partial_{\theta}V^{L,S}_{a}(\theta;\lambda)\big| \, d\theta.
\end{equation}

\end{definition}

When  $V_{a}^{L,S}(\cdot;\lambda)$ has no double zeros,
$$\lim_{\epsilon \rightarrow 0^+} \numzeros_{a,\epsilon}^{L,S}(\lambda) = \numzeros^{L,S}_a(\lambda).$$

To compute $\E\numzeros^{L,S}_a(\lambda ),$ we need to interchange
the $\epsilon \rightarrow 0^+$ limit and integration over $a \in
{\R}^{N^{L,S}(\lambda)}$. This requires showing that
$\numzeros_{a,\epsilon}^{L,S}(\lambda )$ is bounded uniformly in
$\epsilon >0$ by a $\mu$-integrable function in the $a$-variables.

\begin{lemma} \label{LDC}
\label{lem:Leps < ...}  Fix $\lambda >0$ sufficiently large and let
$ \numzeros_{a,\epsilon} ^{L,S}(\lambda)$  be defined as in
\eqref{eq:LLLeps def}.

(i) \, In the case where $\partial \Omega $ is $C^{\infty},$
$$ \sup_{\epsilon \in [0,\epsilon_0]} \numzeros_{a,\epsilon}
^{L,S}(\lambda) \in L^{1} \left( {\mathbb R}^{N^{L,S}(\lambda)}; \, e^{-\|a\|^2/2}\frac{da}{(2\pi)^{N^{L,S}(\lambda)/2}
} \right).$$

(ii) \, In the case where $\partial \Omega$ is $C^{\omega},$ one has  the explicit bound
$$ \sup_{\epsilon \in [0,\epsilon_0]} \numzeros_{a,\epsilon}
^{L,S}(\lambda) \leq C_{\Omega}^{L,S} \lambda,$$
where, $C_{\Omega}^{L,S}>0$ are constants depending only on $\Omega.$
\end{lemma}

\begin{proof}
(i) \,\, The argument for both the long and short ranges are the
same, so without loss of generality, we assume here that $\lambda_j
\in [\lambda, \lambda +1]$. To prove the first part of the Lemma,
only the behaviour in $a\in {\mathbb R}^{N^{S}}$ is important (and
not in $\lambda$). So, we henceforth fix $\lambda >0$ sufficiently
large and suppress dependence of all sets, functions, etc. on
$\lambda$. In the smooth case, the asymptotics in $\lambda$ of the
number of boundary critical points is, to our knowledge, open and
probably rather difficult (see however part (ii) and \cite{TZ1} for
a sharp result in $\lambda$ in the real-analytic case).

First, we note that since $V_{a}^{S}(\theta;\lambda) =
\sum_{\lambda_j \in [\lambda, \lambda +1]} a_j v_j(\theta)$ is
linear in $a \in {\mathbb R}^{N^S}$, the number of boundary critical
points, $C(V_a;\lambda)$, is invariant under scaling by $\|a\| \neq
0$; that is,
\begin{equation} \label{compact}
{\mathcal C}(V_{a}; \lambda) = {\mathcal C}(V_{\omega}; \lambda)
\end{equation}
where, $\omega:= \frac{a}{\|a\|}.$ Thus, we are reduced to proving that
\begin{equation} \label{intsphere}
{\mathcal C}(V_{\omega};\lambda) \in L^{1}( {\mathbb S}^{N^S-1}; d\omega),
\end{equation}
where $d \omega$ is the standard, unit constant curvature  hypersurface measure on ${\mathbb S}^{N^{S}-1}.$

The second point is that by the same argument as in the previous
Lemma \ref{lem:dbl zer prob 0} (see also Corollary \ref{cor:b1, b2
not colinear} and Remark \ref{rem:dt, d2t ind}), with $V_{a}$ (resp.
$\partial_\theta V_{a}$) replaced by $\partial_{\theta} V_{\omega}$
(resp. $\partial_{\theta}^{2} V_{\omega}$) one shows that given the
map $d_{\theta} \Psi_{\lambda}: {\mathbb S}^{N^{S}-1}\times [0,\ell]
\rightarrow {\mathbb R}^{2}$ defined by
$$ d_{\theta} \Psi_{\lambda}:(\omega,\theta) \mapsto (\partial_{\theta} V_{\omega}(\theta;\lambda), \partial_{\theta}^{2} V_{\omega}(\theta;\lambda)),$$
the set $d_{\theta} \Psi_{\lambda}^{-1}(0,0) \subseteq {\mathbb
S}^{N^{S}-1} \times \mathbb{R}$ is a finite union of compact,
$C^{\infty}$ hypersurfaces (here, we use that $V_{\omega}(\theta +
\ell;\lambda) = V_{\omega}(\theta;\lambda) $ for all $\theta \in
{\mathbb R}).$ But then, since ${\mathbb S}^{N^{S}-1}\times {\mathbb
R} \cong {\mathbb R}^{N^{S}} - 0,$  by the generalized
Jordan-Brouwer separation theorem applied to each of the compact,
connected hypersurfaces \cite{A}, it follows that $ ( {\mathbb
S}^{N^{S}-1} \times {\mathbb R}) - \, d_{\theta}
\Psi_{\lambda}^{-1}(0,0)  $ has finitely-many connected components.
Now consider
$$ {\mathcal C}' := \{ \omega \in {\mathbb S}^{N^{S}-1}; \, \exists \theta \in [0,\ell],
\partial_{\theta} V_{\omega}(\theta;\lambda) = \partial_{\theta}^{2} V_{\omega}(\theta;\lambda) =
0 \}.$$ Written another way, ${\mathcal C}' = \pi ( d_{\theta}
\Psi_{\lambda}^{-1}(0,0) )$ where $\pi: {\mathbb S}^{N^{S}-1} \times
{\mathbb R} \rightarrow {\mathbb S}^{N^{S}-1}$ is the smooth
canonical projection map $\pi:(\omega,\theta) \mapsto \omega.$ Since
$\pi$ maps connected sets to connected sets, it follows that
$$ H^{0}( {\mathbb S}^{N^{S}-1} - {\mathcal C}' ) < \infty$$ and clearly,
$\omega ({\mathcal C}') =0.$ Now, make the decomposition ${\mathbb S}^{N^{S}-1}
- {\mathcal C}' = B_{1}\cup \cdots \cup B_{M}; \, M <\infty$ where
the $B_{j}$'s are the (open) connected components (which, implicitly
depend on $\lambda.$) Without any loss in generality, we choose a
point $\omega_{0} \in B_{1}.$  Then, all boundary critical points of
$V_{\omega_0}(\theta;\lambda)$ are simple and we  denote them by
$\theta_0,...,\theta_P$ where $P$ is finite (again, implicitly
dependent on $\lambda$).  Since the following argument is the same
for each of the critical points, we consider the first one: $\theta
= \theta_{0}$.  For $(\theta_0, \omega_0) \in [0,\ell] \times
{\mathbb S}^{N^{S}-1},$ we have by definition
\begin{equation} \label{IFT1}
\partial_{\theta} V_{\omega_0}(\theta_0;\lambda) = 0,
\end{equation}
and since $\omega_0 \in B_1 \subset {\mathbb S}^{N^{S}-1} -
{\mathcal C}',$  it follows that
\begin{equation} \label{IFT2}
\partial_{\theta}^{2} V_{\omega_0} (\theta_0;\lambda) \neq 0.
\end{equation}

Then, by the implicit function theorem, it follows that for $\omega
\in B_1,$ there is a unique $C^{\infty}$ family of solutions
$$(\omega, \theta_0(\omega)) \in B_{1} \times [0,\ell]$$ to \eqref{IFT1}
satisfying $\theta_0(\omega_0) = \theta_0.$ Repeating the same
implicit function theorem argument for the other zeros gives the
existence of smooth families of solutions
$\theta_1(\omega),...,\theta_{P}(\omega)$ to (\ref{IFT1}) for
$\omega \in B_1$ with respective initial conditions
$\theta_1,...,\theta_P$. One can apply the same analysis to each of
the other connected components $B_2,...,B_M$. Let
\begin{equation*}
B_{k,m}=\{\omega\in B_{k}:\: {\mathcal C}(V_{\omega};\lambda) = m
\},
\end{equation*}
so that $$B_{k} = \bigcup\limits_{m=0}^{\infty} B_{k,m}.$$ The
argument above implies that each of the $B_{k,m}$ is open, and thus
also closed (being the complement of $\bigcup\limits_{m'\ne m}
B_{k,m'}$) in $B_{k}$. This implies that for each $m$, either
$B_{k,m}=B_{k}$ or $B_{k,m}=\emptyset$. We conclude that ${\mathcal
C}(V_{\omega};\lambda)$ is constant (and finite) on each of the
(finitely-many) connected components $B_{k}$ of the $N-1$ sphere. It follows that ${\mathcal C}(V_{\omega};\lambda) \in
L^{1}({\mathbb S}^{N^{S}-1};d\omega)$ and so also,
\begin{equation} \label{IFT3}
{\mathcal C}(V_{a};\lambda)  \in L^{1} \left( {\mathbb
R}^{N^{L,S}(\lambda)}; \,
e^{-\|a\|^2/2}\frac{da}{(2\pi)^{N^{L,S}(\lambda)/2} } \right).
\end{equation}

So, given \eqref{IFT3}, it suffices  to prove  that
$\numzeros_{a,\epsilon}^{S}(\lambda) = {\mathcal O}( {\mathcal C}(V_a;\lambda) \, )$ uniformly in $\epsilon >0.$ The domain of integration in the definition \eqref{eq:LLLeps
def} of $\numzeros_{\epsilon}^{S}(\lambda)$ is a finite union
$[0,\ell]=\bigcup\limits_{j=1}^{n} I_{j}$ of disjoint intervals
$I_{j}=[( c_{j},d_{j})]$ each containing at most one critical point.
It is clear that the contribution of each of the $I_{j}$ to the
integral \eqref{eq:LLLeps def} is at most $2$. Moreover, for each
$j$ one of the following holds:
\begin{enumerate}
\item $c_j=0$.

\item $d_j = \ell$.

\item $\partial_{\theta}{V^{S}_{a}}(\theta;\lambda)=0$ for some $\theta\in I_j$.

\item Either ($V_{a}^{S}(c_{j};\lambda) = -\epsilon$ and $V_{a}^{S}(d_{j};\lambda) = \epsilon$) or
($V_{a}^{S}(c_{j},\lambda) = \epsilon$ and $V_{a}^{S}(d_{j};\lambda)
= -\epsilon$).
\end{enumerate}

In the latter case, if $j<n$, then ${V^{S}_{a}}' (\theta;\lambda)
=0$ for some $\theta\in [d_j, c_{j+1}]$. Therefore, $n \le
2+{\mathcal C(V_a;\lambda)+({\mathcal C}(V_a;\lambda)+1)=2{\mathcal
C}(V_a};\lambda)+3$. Thus the integral in \eqref{eq:LLLeps def} is
bounded by
\begin{equation*}
\numzeros_{a,\epsilon}^{S}(\lambda) \le 2 \cdot n \le 2\cdot ( \, 2
{\mathcal C}(V_a;\lambda)+3 \, ) = \mathcal{O}( {\mathcal
C}(V_a;\lambda)).
\end{equation*}
(ii) \,\, Although we will not need the much stronger analytic bound
in Lemma \ref{LDC} (ii) in the current paper, we give the proof here
since in light of recent work
of  Toth and Zelditch \cite{TZ1}, we think it is of independent interest.  From \cite{TZ1} Theorems 1-3, one
has the following bound for the number of boundary critical points
of individual eigenfunctions, $v_{j}$:
\begin{equation} \label{TZ}
{\mathcal C}(v_{j}) \leq C_{\Omega} \lambda_{j}.
\end{equation}
Let $F(\lambda): C^{\infty}(\partial \Omega)\rightarrow
C^{\infty}(\partial \Omega)$ be defined by
$$F(\lambda) f (q)  = \int_{\partial \Omega}  \partial_{\nu_q} G_{0}(q,q';\lambda)  f(q') d\sigma(q')$$
 where $G_{0}(q,q';\lambda) = \frac{i}{4} \text{Ha}^{(1)}_{0}( \lambda |q-q'|)$ is the free outgoing Greens function for the Helmholtz equation in ${\mathbb R}^{2}.$
Roughly speaking, the bound in \eqref{TZ}  follows by
holomorphically continuing both sides of the jumps-equation
\begin{equation} \label{jump}
\partial_{\nu}\phi_{j}|_{\partial \Omega} = -2F(\lambda_j)
(\partial_\nu \phi_{j}|_{\partial \Omega}),
\end{equation} and using a Jensen-type
argument to bound the number of complex (and hence real) zeros in a
complex tube $\partial \Omega_{\C}$ containing $\partial \Omega$ by
the exponential growth exponent of  the holomorphically continued
$F(\lambda)$-kernel. Let $v_{j}^{\C}$ denote the holomorphic
continuation of $v_{j}$ to complex parameter strip $A(\epsilon):[0,\ell] \times [-\epsilon, \epsilon]$ where the holomorphic
continuation of $q$ is  $q^{\C}: A(\epsilon) \rightarrow \partial
\Omega_{\C}$. Then, one has the exponential growth estimate
\cite{TZ1}
\begin{equation} \label{growth}
\sup_{\zeta \in A(\epsilon)} |v_{j}^{\C}(\zeta)| \leq \exp
[ C_{\Omega} \, |\Im \zeta| \, \lambda_{j}]. \end{equation}
Moreover,
$$ V_{a}^{\C} (\zeta) = \sum_{\lambda_j \in [\lambda,\lambda + 1]} a_{j}
v_{j}^{\C}(\zeta)$$ and so by (\ref{growth}) and Cauchy-Schwartz,
\begin{equation} \label{growth2}
| V_{a}^{\C} (\zeta)| \leq  \sqrt{N^{S}(\lambda)} \,\, \|a\| \,\,
\exp [C_{\Omega}\,  |\Im \zeta|\,\lambda]. \end{equation} Since
$N^{S}(\lambda) \sim \lambda,$ Theorem 3 in \cite{TZ1} applied to
$\Phi^{\C}_{a}$ gives
\begin{equation} \label{maingrowth}
{\mathcal C}(V_{a};\lambda) \leq C_{\Omega}  \, \max_{\zeta \in
\partial \Omega_{\C}}  | \, \log | V_{a}^{\C}(\zeta)|  \, |\,  \leq
C_{\Omega}' (  \lambda + | \, \log \|a\| \, |)
\end{equation}
in view of (\ref{growth2}). But by the scaling invariance
(\ref{compact}), it suffices to assume that $\|a\| =1,$ so the last
bound in (\ref{maingrowth}) is ${\mathcal O}_{\Omega}(\lambda).$
Now, just as in (i) one uses that $\sup_{\epsilon}
\numzeros_{a,\epsilon}^{S}(\lambda) = {\mathcal O}( {\mathcal
C}(V_a;\lambda)).$ Again, the same argument works in the long-range
case.

\end{proof}

\begin{proof}[Proof of Proposition \ref{mainprop}]

We use \eqref{eq:Z=ELLL} so that we are to compute the expected
number $\E \numzeros^{L,S}(\lambda)$ of the zeros of $V_{a}^{L,S}
(\cdot;\lambda)$, defined by \eqref{eq:VaL def} and \eqref{eq:VaS
def} on $[0,\ell]$.

By Lemma \ref{lem:dbl zer prob 0}  we can assume that
$V_a^{L,S}(\theta;\lambda)$ has no double zeros so that
\begin{equation*}
\numzeros_{a}^{L,S}(\lambda) = \lim\limits_{\epsilon \rightarrow 0}
\frac{1}{2\epsilon}\int\limits_{0}^{\ell}
\chi_{[-1,1]}\bigg(\frac{V^{L,S}_{a}(\theta;\lambda)}{\epsilon}\bigg)|
\partial_{\theta}{V^{L,S}_{a}}(\theta;\lambda) | d\theta,
\end{equation*}
and so,
\begin{equation*}
\E\numzeros_{a}^{L,S}(\lambda) = \E\lim\limits_{\epsilon \rightarrow 0}
\frac{1}{2\epsilon}\int\limits_{0}^{\ell}
\chi_{[-1,1]}\bigg(\frac{V^{L,S}_{a,\epsilon}(\theta;\lambda)}{\epsilon}\bigg)
|\partial_{\theta}{V^{L,S}_{a}}(\theta;\lambda) | d\theta.
\end{equation*}

Since  by Lemma  \ref{lem:Leps < ...}, for each fixed $\lambda \in {\mathbb R}^{+}$ sufficiently large,  the function
$$ \sup_{\epsilon} \numzeros_{a,\epsilon}^{L,S}(\lambda) \in L^{1} \left( {\mathbb R}^{N^{L,S}(\lambda)}; \, e^{-\|a\|^2/2}\frac{da}{(2\pi)^{N^{L,S}(\lambda)/2}
} \right),$$ by  dominated
convergence, we interchange the order of the integration and limit
and get
\begin{equation*}
\begin{split}
\E\numzeros^{L,S}_{a}(\lambda) &= \lim\limits_{\epsilon \rightarrow 0}
\frac{1}{2\epsilon} \E\int\limits_{0}^{\ell}
\chi_{[-1,1]}\bigg(\frac{V^{L,S}_{a}(\theta;\lambda)}{\epsilon}\bigg)
|\partial_{\theta}{V^{L,S}_{a}}(\theta;\lambda) | d\theta \\&= \lim\limits_{\epsilon
\rightarrow 0}  \int\limits_{0}^{\ell} \E \bigg[\frac{1}{2\epsilon}
\chi_{[-1,1]}\bigg(\frac{V^{L,S}_{a}(\theta;\lambda)}{\epsilon}\bigg)
|\partial_{\theta}{V^{L,S}_{a}}(\theta;\lambda)  |\bigg] d\theta,
\end{split}
\end{equation*}
by Fubini.

We rewrite the last equality as
\begin{equation}
\label{eq:exp N kern} \E\numzeros_{a}^{L,S}(\lambda)=
\lim\limits_{\epsilon \rightarrow 0}  \int\limits_{0}^{\ell}
K_{\epsilon}^{L,S}(\theta;\lambda) d\theta,
\end{equation}
where
\begin{equation} \label{eq:Keps def}
K_{\epsilon}^{L,S}(\theta;\lambda) = \E \bigg[\frac{1}{2\epsilon}
\chi_{[-1,1]}\bigg(\frac{V^{L,S}_{a}(\theta;\lambda)}{\epsilon}\bigg)
|\partial_{\theta}{V^{L,S}_{a}}(\theta;\lambda)  |\bigg].
\end{equation}

Assuming $\lambda$ is fixed, we denote $N=N^{L,S}(\lambda)$. To
compute $K_{\epsilon}^{L,S}(\theta;\lambda)$ for a given $\theta\in
[0,\ell]$, we note that
$$\langle b_{1}^{L,S}(\theta;\lambda),\, a \rangle = V_{a}^{L,S}(\theta;\lambda)$$ and
$$\langle b_{2}^{L,S}(\theta;\lambda),\, a \rangle = \partial_{\theta}{V_{a}^{L,S}}(\theta;\lambda).$$
By  Corollary \ref{cor:b1, b2 not colinear} for $\lambda$ large,
$b_{1}^{L,S}(\theta;\lambda)$ and $b_{2}^{L,S}(\theta;\lambda)$ are
nowhere collinear and so the vectors $\{
b_{1}^{L,S}(\theta;\lambda),\, b_{2}^{L,S}(\theta;\lambda)\}$ can be
extended to a basis
$$\big\{b_{1}^{L,S}(\theta;\lambda),\, b_{2}^{L,S}(\theta;\lambda),\,
b_{3}^{L,S}(\theta;\lambda),\,\ldots,
b_{N}^{L,S}(\theta;\lambda)\big\}$$ of $\R^{N}$ with the property
that $\big\{b_{3}^{L,S}(\theta;\lambda),\,\ldots,\,
b_{N}^{L,S}(\theta;\lambda)\big\}$ is an orthonormal basis of
$\text{span} \{b_{1}^{L,S}(\theta;\lambda),\,
b_{2}^{L,S}(\theta;\lambda)\}^{\perp}$. Let
$B^{L,S}(\theta;\lambda)\in M_{N}(\R)$ be the matrix with row
vectors $b_{k}^{L,S}(\theta;\lambda)$. Then
\begin{equation*}
B^{L,S}(\theta;\lambda) B^{L,S}(\theta;\lambda)^{t}
=\left(\begin{matrix} C^{L,S}(\theta;\lambda) &0\\0
&I_{N-2}\end{matrix}\right),
\end{equation*}
where
\begin{equation*}
C^{L,S}(\theta;\lambda) =\left(\begin{matrix}c_{11}^{L,S}(\theta;\lambda) & c_{12}^{L,S}(\theta;\lambda) \\
c_{21}^{L,S}(\theta;\lambda) &c_{22}^{L,S}(\theta;\lambda)
\end{matrix} \right),
\end{equation*}
with $c_{ij}$ being defined in section \ref{sec:bi, cij def}. In
particular,
$$\det{B}^{L,S}(\theta;\lambda)=\sqrt{\det{C^{L,S}(\theta;\lambda)}}.$$

Writing the Gaussian probability density explicitly in
\eqref{eq:Keps def} yields the formula
\begin{equation}
\label{eq:Keps Gauss expl} K_{\epsilon}^{L,S} (\theta;\lambda)=
\frac{1}{2\epsilon} \int \chi_{[-1,1]}\bigg(\frac{\langle
b_{1}^{L,S}(\theta;\lambda),\, a\rangle}{\epsilon}\bigg) |\langle
b_{2}^{L,S}(\theta;\lambda),\, a \rangle| \exp{(-\frac{1}{2}\|a
\|^2)}\frac{da_1 \cdots da_N}{(2\pi)^{N/2}}
\end{equation}
We change the variables $v= aB$. In the new coordinates, we have
\begin{equation*}
\begin{split}
\|a\|^2 &= a\cdot a^{t} = v B^{L,S}(\theta;\lambda)^{-1}
(B^{L,S}(\theta;\lambda)^{-1})^{t} v^{t} \\&= v \left(\begin{matrix}
C^{L,S}(\theta;\lambda) &0
\\0 &I_{N-2}\end{matrix}\right)^{-1}v^{t} = w_1 C^{L,S}(\theta;\lambda)^{-1}
w_1^t+\|w_2\|^2,
\end{split}
\end{equation*}
where $w_1=(v_1,v_2)$ and $w_2 = (v_3,\ldots,\, v_{N})$, so that
\begin{equation*}
\begin{split}
K_{\epsilon}^{L,S}(\theta;\lambda)
&=\frac{1}{(2\pi)^{N/2}\sqrt{\det{C^{L,S}(\theta;\lambda)}}}\times\\&\times\int\limits_{\R^{N}}
\frac{1}{2\epsilon} \chi_{[-1,1]}\bigg(\frac{v_1}{\epsilon}\bigg)
|v_2| \exp\big(-\frac{1}{2} (w_1
C^{L,S}(\theta;\lambda)^{-1}w_1^{t}+\|w_2\|^2) \big) dv_1 \cdots
dv_{N}
\\& = \frac{1}{(2\pi)\sqrt{\det{C^{L,S}(\theta;\lambda)}}}\int\limits_{\R^2} \frac{1}{2\epsilon} \chi_{[-1,1]}\bigg(\frac{v_1}{\epsilon}\bigg)
|v_2| \exp\big(-\frac{1}{2} w_1 C^{L,S}(\theta;\lambda)^{-1}w_1^{t}
\big)dw_1 \times
\\&\times \int\limits_{\R^{N-2}} \exp\big(-\frac{1}{2} \|w_2\|^2
\big)\frac{dw_{2}}{(2\pi)^{\frac{N-2}{2}}}.
\end{split}
\end{equation*}
Note that the last integrand is just the standard multivariate
Gaussian probability measure, so that the corresponding integral is
just $1$.

Therefore, we get the formula
\begin{equation}
\label{eq:Keps comp}
\begin{split}
K_{\epsilon}^{L,S}(\theta;\lambda)
&=\frac{1}{2\pi\sqrt{\det{C^{L,S}(\theta;\lambda)}}}\times\\&\times\frac{1}{2\epsilon}\int\limits_{-\infty}^{\infty}\int\limits_{-\epsilon}^{\epsilon}
|v_2| \exp\big(-\frac{1}{2}
(v_{1},v_2)C^{L,S}(\theta;\lambda)^{-1}(v_{1},v_2)^{t} \big) dv_{1}
dv_2.
\end{split}
\end{equation}
We wish to apply the dominated convergence theorem again to exchange
the limit and the integral in \eqref{eq:exp N kern}. For this we
estimate $K_{\epsilon}^{L,S}(\theta;\lambda) $ from above in the
following way. Let $E=\{ e_{i} \}$ be any orthonormal basis of
$\R^{N}$, with $e_{1}=\frac{b_{1}}{\|b_{1}\|}$. Then for $a = Ea'$
the integral in \eqref{eq:Keps comp} is, using the invariance of the
Gaussian measure,
\begin{equation*}
\begin{split}
\frac{1}{2\epsilon}\int\limits_{\R^{N}} &\chi_{[-1,1]} \bigg(
\frac{\langle b_{1}^{L,S}(\theta,\lambda), a'_{1} e_{1}
\rangle}{\epsilon}\bigg) \big|\langle b_{2}^{L,S}(\theta,\lambda),
Ea' \rangle \big| \exp\bigg(-\frac{1}{2}\|a'\|^{2}\bigg)
\frac{da'}{(2\pi)^{N/2}}.
\end{split}
\end{equation*}
We have
\begin{equation*}
\langle b_{1}^{L,S}(\theta;\lambda), a'_{1} e_{1} \rangle = a'_{1}
\| b_{1}^{L,S}(\theta;\lambda) \| = a'_{1} \sqrt{c_{11}},
\end{equation*}
so that we integrate for $|a'_{1}| \le
\frac{\epsilon}{\sqrt{c_{11}^{L,S}(\theta;\lambda)}}$, and we have
by the Cauchy-Schwartz inequality
\begin{equation*}
\big|\langle b_{2}^{L,S}(\theta;\lambda), Ea' \rangle \big| \le \|b_{2}^{L,S}(\theta;\lambda)\| \cdot \| Ea'\| \sqrt{c_{22}(\theta;\lambda)} \cdot \| a' \|,
\end{equation*}
so that
\begin{equation*}
\begin{split}
K_{\epsilon}^{L,S}(\theta;\lambda) \le  \frac{
\sqrt{c_{22}^{L,S}(\theta;\lambda)}  }{2\epsilon}
&\int\limits_{-\frac{\epsilon}{\sqrt{c_{11}^{L,S}(\theta;\lambda)}}}^{\frac{\epsilon}{\sqrt{c_{11}^{L,S}(\theta;\lambda)}}}da'_{1}
\int\limits_{\R^{N-1}} \|a'\|e^{-\frac{1}{4}\|a'\|^2} \cdot
e^{-\frac{1}{4}\|a'\|^2} \frac{da'_{2}\ldots da'_{N}}{(2\pi)^{N/2}}
\\ \le C_{1}
\sqrt{\frac{c_{22}^{L,S}(\theta;\lambda)}{c_{11}^{L,S}(\theta;\lambda)}},
\end{split}
\end{equation*}
for some constant $C_{1} > 0$, since $xe^{-\frac{1}{4}x^2}$ is
bounded. Thus we have \begin{equation} \label{bound}
K_{\epsilon}^{L,S} (\theta;\lambda) \leq C_1 \,
\sqrt{\frac{c_{22}^{L,S}(\theta;\lambda)}{
c_{11}^{L,S}(\theta;\lambda)}} \leq C_2\lambda,
\end{equation}
for some constant $C_{2} >0$, where the last estimate in
\eqref{bound} follows from the asymptotics for the $c_{ij}$'s in
section 4 (in the long range case) and section 5 (in the short range
case). Thus by \eqref{eq:exp N kern} and the dominated convergence
theorem,
\begin{equation}
\label{eq:ELLL = int K} \E\numzeros^{L,S}(\lambda)
=\int\limits_{0}^{\ell} K^{L,S}(\theta;\lambda) d\theta,
\end{equation}
where
\begin{equation*}
K^{L,S}(\theta;\lambda) := \lim\limits_{\epsilon \rightarrow 0^+}
K_{\epsilon}^{L,S}(\theta;\lambda).
\end{equation*}
The fundamental theorem of the calculus and \eqref{eq:Keps comp}
then imply that
\begin{equation}
\label{eq:K comp}
\begin{split}
K^{L,S}(\theta;\lambda)
&=\frac{1}{2\pi\sqrt{\det{C^{L,S}(\theta;\lambda)}}}\int\limits_{-\infty}^{\infty}
|v_2| \exp\big(-\frac{1}{2}
(0,v_2)C^{L,S}(\theta;\lambda)^{-1}(0,v_2)^{t} \big) dv_2
\\&=\frac{1}{2\pi\sqrt{\det{C^{L,S}(\theta;\lambda)}}}\int\limits_{-\infty}^{\infty}
|v_2|
\exp(-\frac{1}{2}\frac{c_{11}^{L,S}(\theta;\lambda)}{\det{C^{L,S}(\theta;\lambda)}}v_{2}^2)dv_2
\\&=\frac{1}{2\pi}\frac{\sqrt{\det{C^{L,S}(\theta;\lambda)}}}{c_{11}^{L,S}(\theta;\lambda)}\int\limits_{-\infty}^{\infty}
|z| \exp(-\frac{1}{2} z^2)dz
\\&=\frac{1}{\pi}\frac{\sqrt{c_{11}^{L,S}(\theta;\lambda)c_{22}^{L,S}(\theta;\lambda)-c_{12}^{L,S}(\theta;\lambda)^2}}{c_{11}^{L,S}(\theta;\lambda)},
\end{split}
\end{equation}
since $\int\limits_{\R}\exp(-\frac{1}{2} z^2) |z| dz = 2.$
\end{proof}

\end{document}